\documentclass[a4paper,USenglish,cleveref]{lipics-v2021}

\input{macros}
\nolinenumbers

\newif\iflong
\longtrue

\iflong
    \title{Fast Tendermint: Speeding Up a Foundational Consensus Protocol}
\else
    \title{Brief Announcement: Fast Tendermint: Speeding Up a Foundational Consensus Protocol}
\fi
\titlerunning{Fast Tendermint}

\author{Preston Vander Vos}{Circle Research, San Francisco, USA}{preston.vandervos@circle.com}{0009-0004-7206-5839}{}
\author{Daniel Cason}{Circle Research, Lugano, Switzerland}{daniel.cason@circle.com}{0000-0003-1598-106X}{}

\authorrunning{P. Vander Vos and D. Cason}
\Copyright{Preston Vander Vos and Daniel Cason}



\iflong
    \hideLIPIcs
\else
    \category{Brief Announcement}
	\relatedversiondetails[cite=extended]{Extended version}{https://arxiv.org/abs/2608.13434}
	\supplementdetails[subcategory={Formal Verification}]{Model}
		{https://github.com/circlefin/formal-tendermint/tree/master/fast-tendermint}
\fi

\EventEditors{Ioannis Chatzigiannakis, Andrea Vitaletti, Keren Censor-Hillel, and William K. Moses Jr.}
\EventNoEds{4}
\EventLongTitle{40th International Symposium on Distributed Computing (DISC 2026)}
\EventShortTitle{DISC 2026}
\EventAcronym{DISC}
\EventYear{2026}
\EventDate{November 9--13, 2026}
\EventLocation{Rome, Italy}
\EventLogo{}
\SeriesVolume{397}
\ArticleNo{55}

\acknowledgements{We thank
	Josef Widder and Manuel Bravo for early discussions on the protocol and
	assistance with Quint modeling;
	Hernán Vanzetto for writing the TLA\textsuperscript{+} model and proofs;
	and Alex Ranchal-Pedrosa for surfacing an important problem in an earlier
	version of the pseudo-code.
}

\begin{document}

\maketitle

\begin{abstract}
    Tendermint is among the most widely studied and deployed Byzantine fault-tolerant (BFT) consensus protocols, owing in part to its native leader-rotation mechanism that subsumes complex view changes.
Like most partially-synchronous BFT protocols, Tendermint tolerates $f < n/3$ Byzantine processes and decides in three communication steps.
Motivated by the push for lower-latency blockchains, a recent line of work shows that consensus can be solved in two communication steps when $f < n/5$.
We present \ftm, an adaptation of Tendermint to the $n > 5f$ setting that decides in two communication steps in the good case, while preserving Tendermint's leader-rotation structure.
\ftm collapses Tendermint's prevote and precommit steps into a single voting step and merges the $locked$ and $valid$ state.
We give proofs of agreement, validity, and termination,
and a formal specification in Quint, a modern surface syntax for TLA\textsuperscript{+}, used to model-check the protocol.

\end{abstract}

\section{Introduction} \label{sec:intro}

Consensus is a fundamental abstraction in fault-tolerant distributed systems, at the core of state-machine replication (SMR)~\cite{lamport78, schneider90}; Byzantine fault-tolerant (BFT) consensus~\cite{lamport82} gained renewed attention with the advent of blockchain-based~\cite{bitcoin} SMR.
Inspired by PBFT~\cite{pbft} and the DLS protocol~\cite{dwork88},
Tendermint~\cite{tendermint} has become one of the most widely studied~\cite{tendermint-idiosyncrasies, tendermint-correctness, tendermint-idiosyncrasies, tendermint-design}
and deployed~\cite{cometbft, malachite, cosmos, arc-tendermint} BFT consensus protocols.
Unlike classical BFT protocols, where leader changes are an exceptional event triggered by suspected faults,
Tendermint builds in leader rotation as every round designates a new leader.
A direct consequence is that Tendermint avoids the complex view-change procedures of classical, PBFT-style protocols.
Leader rotation and progress are handled by a single, uniform mechanism, simplifying both reasoning about the protocol and its implementation.
Like most partially-synchronous BFT protocols, Tendermint assumes $f < n/3$
and consequently decides in three communication steps, a tight lower bound at this fault tolerance~\cite{abraham2022}.

Motivated by the ongoing push to make blockchains ever faster,
a recent line of work departs from this setting and instead operates with 
$f < n/5$~\cite{minimmit,alpenglow,fab,bluebottle}
or below the $f<n/3$ threshold~\cite{kudzu, hydrangea}.
At this reduced degree of fault tolerance, BFT consensus can be solved in
two communication steps.
In fact, there are algorithms with good-case latency of two communication steps with
$n \ge 5f-1$~\cite{abraham2022,kuznetsov2021};
however, the additional structural complexity required at this tighter bound appears to have discouraged its adoption in practice.

Instead of creating a new $5f+1$ consensus protocol from scratch, can an existing, battle-tested protocol be adapted to this setting?
We answer affirmatively with Tendermint, which has been deployed in production blockchains~\cite{cosmos, arc-tendermint} and has stood the test of time.
Our contribution is the \ftm protocol, which can decide after one step of voting by receiving $n-f$ precommits for a proposed value.
Tendermint, by comparison, can decide after two steps of voting, $n-f$ prevotes and then $n-f$ precommits for a proposed value.
We accompany it with rigorous safety and liveness proofs,
formalized in TLA\textsuperscript{+}~\cite{tla},
and a model-checked Quint~\cite{quint} specification%
\iflong
	\footnote{Available at
	\url{https://github.com/circlefin/formal-tendermint/tree/master/fast-tendermint}.}.
\fi

\section{Model and Properties} \label{sec:model}

\ftm relies on the same system model as the original Tendermint consensus algorithm~\cite{tendermint}.
We consider a system of $n$ processes that communicate by exchanging messages.
At most $f$ processes may be Byzantine and behave arbitrarily; correct processes follow the protocol.
\ftm requires $n > 5f$ while the original Tendermint requires $n > 3f$.
We assume a public-key infrastructure: protocol messages are signed,
public keys are known to all processes,
and correct processes only consider messages with valid signatures.

We consider the partially-synchronous system model~\cite{dwork88}:
there is a bound $\Delta$ and a Global Stabilization Time ($GST$) such that every message sent after $GST$ is reliably delivered within $\Delta$.
On top of it, we assume the same gossip communication property as Tendermint:
\vspace{-0.7ex}

\begin{itemize}
    \item \emph{Gossip communication:} If a correct process sends a message $m$ at time $t$, then all correct processes receive $m$ before $\max\{t, GST\} + \Delta$.
    Furthermore, if a correct process receives a message $m$ at time $t$, then all correct processes receive $m$ before
    $\max\{t, GST\} + \Delta$.
\end{itemize}

\ftm solves Byzantine consensus with external validity~\cite{cachin01},
a BFT consensus variation motivated by blockchain systems,
defined by the following properties:
\vspace{-0.7ex}

\begin{itemize}
    \item \emph{Agreement (safety):} No two correct processes decide on different values.
    \item \emph{Validity (safety):} Decided values satisfy the application-specific predicate $\valid()$.
    \item \emph{Termination (liveness):} All correct processes eventually decide on a value.
\end{itemize}

\section{\ftm Protocol} \label{sec:protocol}

\noindent
This section walks through \ftm, depicted in \Cref{algo:ftm}, and highlights its differences from Tendermint.
\ftm proceeds as a sequence of consensus instances,
called \emph{heights}, where each height decides a value.
As heights are independent of each other, we focus on the operation of a single height and omit heights from the pseudo-code and messages.
A height consists of one or more \emph{rounds},
always starting from round $0$.
Each round is led by a designated process, the \emph{proposer},
and is an attempt to reach a decision.

A round in \ftm consists of two \emph{steps}: 
\emph{propose} and \emph{precommit}.
Each step has an associated message:
$\Proposal$ and $\Precommit$.
Compared with the original Tendermint, the intermediate \emph{prevote} step and the associated message are dropped.
The core state of a process in Tendermint is stored in two pairs of variables, storing a value and a round: $locked$ and $valid$.
A process \emph{locks} a value when it sends a precommit for it,
then rejects values that are not its locked
one---$locked$ is a safety mechanism.
A value becomes $valid$ when other processes may have locked it;
a correct proposer re-proposes its highest known valid value as it should be an acceptable value---$valid$ is a liveness mechanism.
In \ftm, these two variables are merged into $valid$, which is a pair $(round, value)$.

\begin{algorithm}[htb!]
    \caption{\ftm consensus: tolerates $f$ Byzantine faults with $n > 5f$ processes.}
    \label{algo:ftm}
    \footnotesize 
    \begin{algorithmic}[1]
		\State $round_p := 0$ \label{line:initRound}
		\State $step_p \in \{\propose, \precommit\}$ \label{line:initStep}
		\State $decision_p := \nil$ \label{line:initDecision}
		\State $valid_p := (-1, \nil)$ \label{line:initValid}
			\Comment{$(round,\ \id(value))$}
        \HalfStatex

        \Function{StartRound}{round} \label{line:startRound}
            \State $round_p \assign round$ \label{line:setRound}
            \State $step_p \assign \propose$ \label{line:setStepPropose}
            \If{$\coord(round_p) = p$} \label{line:ifProposer}
                \If{$round_p > 0$} \label{line:ifRoundPositive}
                    \State \Call{WaitForValid}{$\timeoutPrecommit(round_p)$}  
                \EndIf
                \If{$valid_p.value = \nil$} \label{line:ifNotValid}
                    \State $proposal \assign \getValue()$ \label{line:getNewValue} \Comment{Fresh proposals carry a full value $v$}
                \Else
					\State $proposal \assign valid_p.value$ \label{line:proposeValidValue} \Comment{Re-proposals carry a value identifier \id(v)}
                \EndIf
                \State \algbroadcast $\li{\Proposal, round_p, proposal, valid_p.round}$ \label{line:broadcastProposal}
            \Else
                \State \algschedule \Call{OnTimeoutPropose}{$round_p$} \algafter $\timeoutPropose(round_p)$ \label{line:scheduleTimeoutPropose}
            \EndIf
        \EndFunction
		\HalfStatex

        \Upon{$\li{\Proposal, round_p, v, -1}$ \algfrom $\coord(round_p)$ \algwhile $step_p = \propose$}\label{line:recvProposalFresh}
            \If{$\valid(v) \wedge (valid_p.round = -1 \vee valid_p.value = \id(v))$} \label{line:checkValidUnlocked}
                \State \algbroadcast $\li{\Precommit, round_p, \id(v)}$ \label{line:broadcastPrecommitFresh}
            \Else
                \State \algbroadcast $\li{\Precommit, round_p, \nil}$ \label{line:broadcastPrecommitFreshNil}
            \EndIf
            \State $step_p \assign \precommit$ \label{line:setStepPrecommitFresh}
        \EndUpon
        \HalfStatex

        \Upon{$\li{\Proposal, round_p, \id(v), vr}$ \algfrom $\coord(round_p)$ \algwith $0 \leq vr < round_p$ \algand \\ $2f+1$ $\li{\Precommit, vr, \id(v)}$ \algwhile $step_p = \propose$} \label{line:recvProposalReproposal}
            \If{$valid_p.round \leq vr \vee valid_p.value = \id(v)$} \label{line:checkValidReproposal}
                \If{$valid_p.round \leq vr$} \label{line:checkValidRoundLEvr}
                    \State $valid_p \assign (vr, \id(v))$ \label{line:updateLockReproposal}
                \EndIf
                \State \algbroadcast $\li{\Precommit, round_p, \id(v)}$ \label{line:broadcastPrecommitReproposal}
            \Else
                \State \algbroadcast $\li{\Precommit, round_p, \nil}$ \label{line:broadcastPrecommitReproposalNil}
            \EndIf
            \State $step_p \assign \precommit$ \label{line:setStepPrecommitReproposal}
        \EndUpon
        \HalfStatex

        \Upon{$2f+1$ $\li{\Precommit, round_p, \id(v)}$ \algwith $round_p > valid_p.round$} \label{line:uponObservation}
            \State $valid_p \assign (round_p, \id(v))$ \label{line:updateLockObservation}
        \EndUpon
        \HalfStatex

        \Upon{$n-f$ $\li{\Precommit, r, *}$ for the first time \algwith $r \geq round_p$}\label{line:uponPrecommitQuorum}
            \State \algschedule \Call{OnTimeoutPrecommit}{$r$} \algafter $\timeoutPrecommit(r)$ \label{line:scheduleTimeoutPrecommit}
        \EndUpon
		\HalfStatex

        \Upon{$\li{\Proposal, r, v, -1}$ \algfrom $\coord(r)$ \algand $n-f$ $\li{\Precommit, r', \id(v)}$} \label{line:uponDecide}
            \State $decision_p \assign v$ \label{line:decide}
        \EndUpon
		\HalfStatex

        \Function{WaitForValid}{$timeout$} \label{line:waitForValid} \Comment{Run by the proposer of rounds $>0$}
            \While{$(valid_p.round < round_p - 1)~\wedge~!timeout.elapsed()$}
                \If{$2f+1$ $\li{\Precommit, r, \id(v)}$ \algwith $round_p > r > valid_p.round$}
                    \State $valid_p \assign (r, \id(v))$ \label{line:updateLockWait}
                \EndIf
            \EndWhile
        \EndFunction
        \HalfStatex

        \Function{OnTimeoutPropose}{$round$} \label{line:onTimeoutPropose}
            \If{$round = round_p \wedge step_p = \propose$} \label{line:checkRoundStepPropose}
                \State \algbroadcast $\li{\Precommit, round_p, \nil}$ \label{line:broadcastPrecommitTimeoutNil}
                \State $step_p \assign \precommit$ \label{line:setStepPrecommitTimeout}
            \EndIf
        \EndFunction
        \HalfStatex

        \Function{OnTimeoutPrecommit}{$round$} \label{line:onTimeoutPrecommit}
            \If{$round \geq round_p \wedge decision_p = \nil$} \label{line:checkRoundTimeoutPrecommit}
                \If{$2f+1$ $\li{\Precommit, round, \id(v)}$ \algwith $round > valid_p.round$} \label{line:checkValidTimeout}
                    \State $valid_p \assign (round, \id(v))$ \label{line:updateLockTimeout}
                \EndIf
                \State \Call{StartRound}{$round + 1$} \label{line:nextRound}
            \EndIf
        \EndFunction
    \end{algorithmic}
\end{algorithm}

\myparagraph{Ordinary operation:}
Processes start a height from round $r = 0$,
with empty $valid = (-1, \nil)$.
Process $\coord(r)$ calls $\getValue()$ to obtain a value $v$ to propose,
then broadcasts $\li{\Proposal, r, v, -1}$ (line~\ref{line:broadcastProposal}), where a broadcast also delivers the message to the sender.
Upon receiving this proposal, a process  applies the 
\emph{\freshProposalRule} (line~\ref{line:recvProposalFresh}): provided $\valid(v)$ holds, and since $valid$ is empty,
it broadcasts $\li{\Precommit, r, \id(v)}$, where $\id(v)$ is a short unique identifier of $v$.
Once the proposal and $n-f$ matching precommits for $\id(v)$ have been received, the \emph{\decisionRule} (line~\ref{line:uponDecide}) fires and the process decides $v$.

\myparagraph{Faulty proposer:}
As in Tendermint, non-proposer processes schedule a timeout when starting round $r$ (line~\ref{line:scheduleTimeoutPropose}).
If it expires before a proposal is received, the process
broadcasts $\li{\Precommit, r, nil}$.
A faulty proposer can also equivocate, proposing different values to different processes; 
this may produce conflicting votes.
To handle both situations, whenever a process receives $n-f$ precommits for any value, the \emph{\precommitQuorumRule} (line~\ref{line:uponPrecommitQuorum}) schedules a timeout.
If it expires before a decision is reached, the process starts the next round $r+1$.

\myparagraph{Valid values:}
A value is valid in a round if no other value can be decided in that round.
In Tendermint, a value is valid if it receives $n-f$ prevotes;
in \ftm, if it receives $2f+1$ precommits from the same round,
as proven in \Cref{lemma:value-lock}.
When this happens for $\id(v)$ in the process's \emph{current round},
the \emph{\observationRule} (line~\ref{line:uponObservation})
sets $valid = (round_p, \id(v))$.
It can also happen for a previous round $vr < round_p$,
as part of the \emph{\reProposalRule} (line~\ref{line:updateLockReproposal}).
A process with a non-empty $valid$ value rejects, by default, any received proposal
for a different value, broadcasting a $\li{\Precommit, r, nil}$;
the exception to this rule is described next.

\myparagraph{Re-proposals:}
A proposer with a non-empty $valid$ value re-proposes it
by broadcasting $\li{\Proposal, r, \id(v), vr}$,
where $\id(v)$ is the valid value and $vr$ the associated valid round.
At a receiving process $q$, the \emph{\reProposalRule}
(line~\ref{line:recvProposalReproposal}) is activated
provided that the $2f+1$ precommits from round $vr$
supporting the re-proposal were received.
They can \emph{override} a conflicting $valid_q$
if $vr \geq valid_q.round$, in which case $q$ broadcasts $\li{\Precommit, r, \id(v)}$.

\myparagraph{Round skipping:}
In Tendermint, late processes jump to a higher round $r'$
after receiving $f+1$ messages from round $r'$.
In \ftm, for the sake of safety,
a process must record the valid value of round $r$, if any,
before it moves to round $r+1$ (lines~\ref{line:checkValidTimeout}--\ref{line:updateLockTimeout}).
For this reason, the \precommitQuorumRule, with an $n-f$ quorum,
is the only allowed path to skip rounds.

\myparagraph{Remarks:}
\ftm makes some additional optimizations.
Re-proposals carry the value identifier $\id(v)$ rather than the full value $v$,
so only the original fresh proposal carries $v$.
This allows $valid$ to be updated from $2f+1$ precommits alone,
not requiring a proposal for $v$.
The $\valid()$ predicate is only evaluated in the \freshProposalRule;
\Cref{thm:validity} explains how validity is ensured in other rules.
The \textsc{WaitForValid()} method is introduced for the sake of liveness,
as detailed in \Cref{lemma:valid-round}.
Before proposing in a round $r > 0$, the proposer runs it for an amount of
time so as to learn 
valid values from rounds $r' \le r-1$, if any.

\section{Correctness Proofs} \label{sec:proofs}

In the protocol, precommit messages and the $valid$ state carry value identifiers $\id(v)$ rather than full values $v$.
We assume $\id$ is collision-resistant, so $\id(v) = \id(w)$ implies $v = w$.

\begin{lemma} \label{lemma:small-intersection}
    Any two sets of $n-f$ and $2f+1$ processes share at least one correct process.
\end{lemma}
\begin{proof}
    The two sets intersect in at least $(n - f) + (2f+1) - n = f+1$ processes.
    Since at most $f$ processes are faulty, at least one process in the intersection is correct.
\end{proof}

\begin{lemma} \label{lemma:value-lock}
    If $n-f$ processes precommit for $\id(v)$ in round $r$ ($v$ can be decided in round $r$),
    then no correct process sets $valid = (r', \id(w))$ for any $\id(w) \neq \id(v)$ and $r' \geq r$.
\end{lemma}
\begin{proof}
    \iflong
    We prove the claim by induction on $r'$, where $r'$ is the highest round that any correct process has entered.
    Since no correct process has entered any round $> r'$, at most $f$ precommits from faulty processes exist for rounds $> r'$.
    Updates to $valid$ in \ftm require $2f+1$ precommits in a round, so we can safely disregard rounds $> r'$ and prove the lemma by showing that no $\id(w) \neq \id(v)$ receives $2f+1$ precommits in any round $r' \geq r$.

    \emph{Base case 1 ($r' = r$):}
    By \Cref{lemma:small-intersection}, the $n-f$ processes that precommit for $\id(v)$ and the $2f+1$ processes that precommit for $\id(w)$ share at least one correct process.
    Correct processes only send one precommit per round, so $\id(w) = \id(v)$.

    \emph{Base case 2 ($r' = r+1$):}
    Denote by $C$ the set of correct processes that join round $r+1$.
	Processes in $C$ receive $n-f$ precommits for round $r$ to trigger the \precommitQuorumRule. 
    Since at most $f$ correct processes did not precommit for $\id(v)$ in round $r$, and at most $f$ processes are faulty, at least $n-3f > 2f$ precommits received by every process in $C$ are for $\id(v)$.
	Immediately before any process $p$ in $C$ executes \OnTimeoutPrecommit($r$) to join round
	$r+1$, $valid_p.round \leq round_p \leq r$.
	If $valid_p.round = r$, it was set by the observation rule
	(line~\ref{line:uponObservation}),
	and $valid_p.value = \id(v)$, as no $\id(w) \neq \id(v)$
	receives $2f+1$ precommits in round $r$ (Base case 1).
	Otherwise,
	$valid_p.round < r$ (e.g., due to round skipping), the condition in
	line~\ref{line:checkValidTimeout} holds and $p$ sets $valid_p = (r,
	\id(v))$.
	In both cases, $p$ enters round $r+1$ with $valid_p = (r, \id(v))$;
	this is valid for all correct processes joining round $r+1$.
	There are thus only faulty processes that can precommit for $\id(w) \neq \id(v)$ in round
	$r+1$.

    \emph{Inductive step ($r' > r+1$):}
    Assume for all rounds $r''$ with $r+1 \leq r'' < r'$, no $\id(w) \neq \id(v)$ received $2f+1$ precommits in round $r''$.
    We prove the lemma for round $r'$.

    By the induction hypothesis, no process in $C$ updated its $valid$ to $\id(w) \neq \id(v)$ before round $r'$.
    Therefore, processes in $C$ that precommit in round $r'$ still have $valid.value = \id(v)$ and do not precommit for $\id(w) \neq \id(v)$.
    There are at most $f$ correct processes not in $C$ and $f$ faulty processes that can precommit for $\id(w)$ in round $r'$, totaling $2f < 2f+1$.

    \fi
\end{proof}

\begin{theorem} \label{thm:safety}
    \ftm satisfies Agreement.
\end{theorem}
\begin{proof}
    Suppose correct processes $p$ and $q$ decide $v$ and $w$, respectively.
    By the \decisionRule
    , $p$ (resp. $q$) observed $n-f$ precommits for $\id(v)$ (resp. $\id(w)$) in some round $r_p$ (resp. $r_q$).
    Without loss of generality, $r_p \leq r_q$.
    Since $\id(v)$ received $n-f$ precommits in round $r_p$, by \Cref{lemma:value-lock},
    no value other than $\id(v)$ can receive $2f+1$ precommits in any round $r \geq r_p$.
    This includes round $r_q$, where $q$ observed $n-f > 2f+1$ precommits for $id(w)$.
    Consequently, $\id(w) = \id(v)$.
    By collision-resistance of $\id$, $w = v$.
\end{proof}

\begin{theorem} \label{thm:validity}
    \ftm satisfies Validity.
\end{theorem}
    
\begin{proof}
    The \freshProposalRule 
    requires the application-specific $\valid(v)$ predicate to attest the validity of $v$ in order to produce a precommit for $\id(v)$.
    The \reProposalRule 
    and the \decisionRule 
    do not check
    for validity. They assume that $v$ is valid based on $2f+1$ previous precommits or $n-f$ precommits for $\id(v)$, respectively.

    Let $r$ be the earliest round in which any correct process precommits for $\id(v)$.
    Since the \reProposalRule requires $2f+1$ previous precommits for $\id(v)$
    and at most $f$ processes are faulty, it cannot be triggered.
    The precommits in round $r$ from correct processes were therefore broadcast by the \freshProposalRule, establishing $\valid(v)$.
\end{proof}

\begin{lemma} \label{lemma:valid-round}
Let $p = \coord(r)$ be a correct process
and $vr$ be $valid_p.round$ when $p$ broadcasts its proposal.
If $p$ started round $r > 0$ at time $t > GST + \timeoutPrecommit(r-1)$,
no correct process started a round $>r$ by $t$,
and $\timeoutPrecommit(r) > 2\Delta$, then $vr$ is at least the maximum
$valid_q.round$ among every correct process $q$ at the time $q$ joins round $r$.
\end{lemma}

\begin{proof}
\iflong
    Since $valid_q.round$ is updated
	to the current round via the \observationRule (line~\ref{line:uponObservation}),
    to a previous round via the \reProposalRule (line~\ref{line:updateLockReproposal}),
	or to the round for which the \precommitQuorumRule was observed (line~\ref{line:checkValidTimeout}),
    the round it records is never higher than the round $q$ is in or is leaving.
    Hence $valid_q.round \leq r-1$ when any correct process $q$ enters round $r$.
    
    Correct processes start round $r > 0$ upon the expiration of
    $\timeoutPrecommit(r-1)$.
    Since $p$ scheduled it after $GST$, from the gossip communication property, every correct process $q$ will schedule the timeout
    (line~\ref{line:scheduleTimeoutPrecommit}) at most $\Delta$ after $p$, then enter round $r$ by $t + \Delta$.
    
	The last opportunity for a correct process $q$ to update $valid_q$
	before entering round $r$ is 
	when \OnTimeoutPrecommit($r-1$) is executed,
	namely, by $t + \Delta$.
	The same messages that lead to this update are received by $p$ by $t +
	2\Delta$.
	Since \textsc{WaitForValid()} 
	runs for up to $\timeoutPrecommit(r) > 2\Delta$,
    the proposer $p$ updates $vr$ to at least the maximum of $valid_q.round$ among every correct process $q$, before broadcasting its proposal for round $r$.
\fi
\end{proof}

\begin{lemma} \label{lemma:decision-in-good-round}
If $r > 0$ is a round such that:
(1) the first correct process enters round $r$ at time $t > GST + \timeoutPrecommit(r-1)$,
(2) $p = \coord(r)$ is a correct process, 
(3) $\timeoutPropose(r) > 2\Delta + \timeoutPrecommit(r)$,
and (4) $\timeoutPrecommit(r) > 2\Delta$,
then all correct processes decide in round $r$.
\end{lemma}

\begin{proof}
\iflong

	Since the system is past $GST$, by the gossip communication property, every
	correct process enters round $r$ by $t+\Delta$.
	This includes the proposer: $p$ broadcasts its proposal by $t + \Delta + \timeoutPrecommit(r)$.
	As every correct process joins round $r$ between $t$ and $t + \Delta$,
	condition~(3) ensures $\timeoutPropose(r)$ has not elapsed when,
	within $\Delta$ from its broadcast, every correct process receives $p$'s
	proposal.
	So, before any correct process receives the proposal for round $r$,
	no correct process has:
	(i) broadcast a round-$r$ precommit,
	(ii) started any round $>r$,
	and (iii) updated its $valid$---except for the proposer $p$.

	As a result, \Cref{lemma:valid-round} applies and the valid round $vr$
	carried in $p$'s proposal is at least the maximum $valid_q.round$ across
	all correct processes $q$ when they joined round $r$.
	Hence, every correct process $q$ accepts $p$'s proposal for $v$---either
	because $vr \ge valid_q.round$ or because $valid_q.value = \id(v)$---and
	broadcasts a precommit for $\id(v)$.
	The time from when the first correct process precommits until all correct
	processes have received all $n-f$ correct precommits is no longer than
	$2\Delta$ ($< \timeoutPrecommit(r)$ by condition (4)), from the gossip
	communication property.
	It follows that every correct process decides $v$ in round $r$.
\fi
\end{proof}

\begin{theorem} \label{thm:termination}
\ftm satisfies Termination.
\end{theorem}

\begin{proof}
\Cref{lemma:decision-in-good-round} establishes sufficient conditions
under which a round $r$ yields a decision.
It remains to show that these conditions are met within a bounded
interval after $GST$.

If a correct process triggers the \decisionRule 
at time $t$, by the gossip communication property every correct process
receives the same messages  and also triggers it by time $\max\{t, GST\} + \Delta$.
Note that this rule is not restricted by the process's current round.

If no decision is reached prior to $GST$, \Cref{lemma:decision-in-good-round} requires 
(1) $GST$ to be reached before the first correct process to enter round $r$ schedules $\timeoutPrecommit(r-1)$, and
(2) a correct process $p$ to be $\coord(r)$.
Since processes advance past unsuccessful rounds
and the proposer role rotates among all processes,
after $GST$ there are infinitely many rounds that satisfy conditions~(1) and~(2).

Finally, the adoption of adaptive timeouts, which increase over rounds,
ensures that the timing conditions (3) and (4) 
of \Cref{lemma:decision-in-good-round} 
are eventually satisfied, even if $\Delta$ is unknown.
Condition~(3) follows from~(4) by assuming $\timeoutPropose(r) > 2\,\timeoutPrecommit(r)$.
\end{proof}

\section{Related Work}
\label{sec:rel-work}

\ftm is one of several BFT protocols that aim to decide after a single voting round.
Kudzu~\cite{kudzu} operates with $n \geq 3f + 2t + 1$ processes, committing in two rounds when at most $t$ of them are faulty and in three rounds under the full $f$ Byzantine budget.
With $t = f$, its fast path is similar to that of \ftm.
Hydrangea~\cite{hydrangea} parameterizes its fault budget with $n \geq 3f + 2c + k + 1$ and commits in two rounds when at most $\lfloor (c+k)/2 \rfloor$ processes are faulty, where $c$ bounds crash faults and $k \geq 0$ is a tunable parameter.
Retaining a failure-optimal slow path, as both protocols do, is an open question for \ftm.

Alpenglow~\cite{alpenglow} claims a two-round commit while simultaneously tolerating up to $20\%$ Byzantine and $20\%$ crash faults, yet this combined guarantee has been refuted~\cite{hydrangea}.
Moreover, once its erasure-coded dissemination layer is accounted for, end-to-end finality requires three communication rounds rather than two~\cite{bluebottle}.

Minimmit~\cite{minimmit} and BlueBottle~\cite{bluebottle}, like \ftm, achieve two-round commits with $n > 5f$ without a slow path.
Minimmit's view progression rule is similar to \ftm's \observationRule and may allow for further speed-up.
BlueBottle belongs to a recent (relative to the original Tendermint) line of BFT consensus protocols that construct and decide over a directed acyclic graph (DAG) rather than a single leader's chain.

Two-round finality is in fact achievable at the tighter bound $n \geq 5f-1$~\cite{abraham2022, kuznetsov2021} by identifying and excluding faulty leaders.
However, the added logic required to achieve this bound, in addition to the complex view changes of these protocols, does not make the trade-off currently beneficial for real-world deployments.

\section{Conclusion} \label{sec:conclusion}

\ftm adapts Tendermint to the $n > 5f$ setting and decides in two communication
steps in the good case while preserving Tendermint's per-round leader rotation that subsumes complex view changes.
The protocol drops the prevote step and merges $locked$ and $valid$ states into a single variable, yielding a minimal modification to a long-deployed consensus protocol.
Future work includes its full implementation and experimental evaluation.

\bibliographystyle{plainurl}
\iflong
    \bibliography{references}
\fi

\end{document}